\documentclass[11pt]{article}
\usepackage{amsthm,amsmath,amsthm,amsfonts,amssymb,eufrak,multirow,url}
\usepackage{color, epsfig}

\usepackage{fullpage}

\usepackage{hyperref}
\hypersetup{backref, colorlinks=true, citecolor=blue, linkcolor=blue}

\newcommand{\norm}[1]{||#1||}

\newtheorem{thm}{Theorem}[section]

\newtheorem{lem}[thm]{Lemma}

\usepackage[numbers]{natbib}
\title{Achieving Both Valid and Secure Logistic Regression Analysis on Aggregated Data from Different Private Sources}
\author{Stephen E. Fienberg\thanks{Department of Statistics, Machine Learning Department and Cylab, Carnegie Mellon University, Pittsburgh, PA. \tt{mailto:fienberg@stat.cmu.edu}} \hspace{3pt}and Robert J. Hall\thanks{Department of Statistics and Machine Learning Department, Carnegie Mellon University, Pittsburgh, PA. \tt{mailto:rjhall+@cs.cmu.edu}} \hspace{1pt}    and Yuval Nardi\thanks{Faculty of Industrial Engineering and Management, Technion - Israel Institute of Technology, Haifa, Israel. \tt{mailto:ynardi@ie.technion.ac.il}}}
\date{}

\begin{document}

\maketitle

\begin{abstract}
Preserving the privacy of individual databases when carrying out statistical calculations has a long history in statistics and had been the focus of much recent attention in machine learning  In this paper, we present a protocol for computing logistic regression when the data are held by separate parties without actually combining information sources by exploiting results from the literature on multi-party secure computation.  We provide only the final result of the calculation compared with other methods that share intermediate values and thus present an opportunity for compromise of values in the combined database.  Our paper has two themes:  (1) the development of a secure protocol for computing the logistic parameters, and a demonstration of its performances in practice, and (2) and amended protocol  that speeds up the computation of the logistic function.   We  illustrate the nature of the calculations and their accuracy using an extract of data from the Current Population Survey divided between two parties.\\

{\bf Keywords:}  Distributed analysis; Logistic regression; Privacy-preserving computation; Secure multiparty computation.

\end{abstract}

\section{Introduction}

Privacy concerns are becoming more and more acute, especially in the digitized world where new supercomputers with an increasing processing capacities appear almost every day. These new machines together with impressive new technologies make the process of data collection, data storing and data analysis as easy as ever. This ``ease of use,'' may be manipulated by untrustful elements, whose aim is to deliberately cause harm by, for example, identifying and exposing sensitive data. It is the goal of privacy preserving methods to prevent or at least  lessen the chances of such harmful actions from happening. In this paper we present a novel way to achieve the goal when a certain statistical analysis is required.

Preserving the privacy of individual databases when carrying out statistical calculations has a long history in statistics and had been the focus of much attention in machine learning, e.g., see~\cite{aggarwal}. Once data are merged across sources, however, privacy becomes far more complex and a number of privacy issues arise for the linked individual files that go well beyond those that are considered with regard to the data within individual sources. When the goal is the production of the results of some statistical calculation, such as a regression analysis, c.f. Karr et al.~\cite{securereg-jcgs05,securereg-springer06}, we can often exploit results from the cryptography literature, borrowing tools  such as secure multi-party computation, e.g., see~\cite{lindell_pinkas_1,ppdm_book}. Secure multi-party protocols are concerned with distributed computation where each participating party, holding a private input, learns nothing but the result (see Section \ref{sec:smpc}).

This paper has two main themes. In both themes we conceptualize the existence of a single combined database containing all of the information for the individuals in the separate databases and for the union of the variables.    We propose an approach that gives full statistical calculation on this combined database without actually combining information sources, see~\cite{lindell_pinkas_2,lindell_pinkas_1}.  We focus mainly on logistic regression, but our methods and tools are essentially adaptable to other statistical models, as indicated in Section \ref{sec:extension}.  Our approach provides only the final result of the calculation compared with other methods that share intermediate values and thus present an opportunity for compromise of values in the combined database, c.f.~\cite{fien_secure,fien_secure2}. We remark that our problem differs from the one studied by Chaudhuri and Monteleoni~\cite{Chaudhuri} using differential privacy, since they are concerned with information leakage by the output of the computation, whereas we are concerned with leakage from the computation itself !

The first theme is the development of a novel approach to perform the calculations required for fitting logistic regression models when the data are distributed among several parties. In our settings the parties are unwilling or are simply forbidden (by law regulations) to share their respective data. They acknowledge the fact that pooling their private data into a conceptual global database, and running the logistic regression on the pooled data, rather than on their own data, can only lead to a better statistical analysis. We develop a secure protocol to compute the maximum likelihood estimates of the logistic parameters. Throughout the paper we make repeated use of what is known as random secret  sharing, which enables us to keep intermediate parameter values secret. The first theme aims at performing the required calculations by using operations which are restricted to a linear algebra type. Note that the fitting process requires computing the logistic function which is highly non-linear.
In principle, we may perform  any computation securely, by making use of Yao's general protocol \cite{yao82}.  Nonetheless, this is in essence a theoretical construction which will often be inefficient for large computations \cite{ppdm_book}. Instead, we craft a specially designed approximation to the logistic function, which can be securely evaluated using the machinery of random shares and Yao's millionaire protocol.

We establish the theoretical validity of the secure protocol for computing the logistic parameters, and show its performances in practice. In high dimensional problems with large number of cases our protocol may require faster computing resources. This is mainly because our approximation requires computing the predicate ``greater-than,'' which may take many encryptions. Indeed, evaluating this predicate by a reduction to Yao's protocol takes roughly $O(b)$ encryptions where $b$ is the number of bits used to represent the numbers (this becomes dauntingly large due to the secret sharing scheme).

This leads us to the second theme, which tries to amend the protocol is a way that speeds up the computation of the logistic function. The main idea here is to avoid special circuit sub-protocols, such as  Yao's protocol. To that end, we show that we can perform the fitting process  using only sums and products. The advantage to this  is that these computations are very well studied primitives in secure multiparty computation and thus we can instantiate our method in a different secure multiparty computation scheme (e.g., \cite{ppdm_book, goldreich}), depending on the security demands of the data holders. We propose to approximate the vector of logistic function values, by repeatedly Taylor expanding around the current value and stepping along the gradient. Operations other than sums and products, are not needed here.   In principle, the approach represents the logistic function as the solution to a ordinary differential equation, then applies Euler's method to approximate the solution. As with the first theme, we show that we can make  the approximation arbitrarily accurate, at the expense of computational efficiency, and we present an illustrative empirical result.

We close the introduction with a brief description of logistic regression, mainly for the purpose of setting notation.
Logistic regression is used for predicting binary outcomes or class membership given a set of explanatory variables or predictors. We can use the fitted model  to predict class membership for a newly obtained record consisting of only the values of  the predictors. The basic framework of logistic regression treats binary responses $y_1, \ldots, y_n$ as realizations of $n$ independent Bernoulli random variables, $Y_1, \ldots, Y_n$, whose mean depends on a set of predictors $x_i\in\mathbb{R}^d$, as follows:
\begin{equation}
\mathbb{E}Y_i = \sigma( x_i^T\beta) \; ,
\end{equation}
where $\sigma(a) = (1+\exp(-a))^{-1}$, is the sigmoid (or the logistic) function, and $\beta$ is a $d$-dimensional parameter vector. This makes the log odds, $\log(\mathbb{E}Y_i/(1-\mathbb{E}Y_i))$, linear in the predictors.

A standard method for computing the maximum likelihood estimates of $\beta$ is Newton-Raphson's  method, since closed form expressions do not exist. The fitting process  requires the user to supply the log-likelihood function associated with logistic regression, along with its first two derivatives.  Suppressing dependence on the data and vector of parameters, we let $\ell$ be the  log-likelihood, i.e., $\ell=\sum_i\{y_ix_i^T\beta - \log(1+e^{x_i^T\beta})\}$. We also put on record the first two derivatives:


\begin{equation}\label{grad-hess}
\nabla\ell = \sum_i\{x_i y_i - x_i \sigma(x_i^T\beta)\} \quad ,\quad
\nabla^2 \ell = -\sum_i{\sigma(x_i^T\beta)(1-\sigma(x_i^T\beta))x_ix_i^T} \; .
\end{equation}


The gradient and the Hessian are assembled together to produce an estimate of the logistic parameters through the iterative process:  
\begin{equation}\label{eq:nr}
\beta_{t+1}=\beta_t - (\nabla^2 \ell)^{-1} \nabla \ell \; .
\end{equation}
Our protocol will be structured in rounds, where each round corresponds to an iteration of Newton's method (\ref{eq:nr}) followed by a convergence check.  Each round involves a loop through all the cases $x_i$ to compute the contribution to the gradient and Hessian.  
We keep intermediate values of $\beta_t$ unshared between the parties. This is made possible by representing $\beta_t$ as random shares (see Section \ref{sec:blocks}).

The remainder of the paper is organized as follows. Section \ref{sec:smpc} presents the multi-party setup.  In section \ref{sec:blocks} we provide several sub-protocols which we will need. Sections \ref{sec:protocol1} and \ref{sec:protocol2} describe our protocol and an approach for speeding up the calculation involved, respectively. Section \ref{sec:security} describes implementation details. Section \ref{sec:experiment}  illustrates aspects of the computation on an extract of data from the Current Population Survey divided between two parties. Section  \ref{sec:extension} discusses possible extensions.  We defer all technical details  to Appendices \ref{sec:validity_1} and \ref{sec:validity_2}.


\subsubsection*{Setting}

We let $X$ denote the $n\times d$ design matrix, and $y$ the $n$-dimensional response vector. We assume the presence of $P \geq 2$ parties who are interested in computing logistic regression on the total of their data.  We suppose that the union of the parties data corresponds to the $X$ and $y$ of the logistic regression.  In particular, we suppose that party $j$ holds onto the pair $(X_j, y_j)$ with $X_j \in \mathbb{R}^{n\times d}$ and $y_j \in \{0,1\}^n$, where $X_j$ is the $j^{th}$ party design matrix, and $y_j$ is her (binary) response vector.

In this work we consider a setting where each party has an ``additive share'' of the dataset.  That is, $\sum_j X_j = X$ and $\sum_j y_j=y$ where $X$ and $y$ correspond to the design matrix and response vector of the combined data on which the logistic regression is performed.  This subsumes all the partitioning schemes for the database (e.g., vertical and horizontal partitioning which are the cases considered in \cite{ppdm_book}) as in these cases for each element, one party holds the value and the remaining parties hold zero.  Furthermore this setup is applicable in a case where parties may have overlapping data, and the logistic regression is to be learned by using a linear function of the overlapping data (e.g., a weighted average) as a kind of measurement error model.  We suppose that the union of the individual data sets gives the complete data.  In cases where some data are missing, we can apply a privacy preserving imputation method such as in \citet{jagannathan}  as a preprocess, and then run our protocol.

We note that our method is general in the sense that it is applicable to every partitioning scheme, but it is clearly possible to treat specific cases such as vertically partitioned data  with more efficient specialized protocols.

\section{Secure Multi-Party Computation}\label{sec:smpc}

Ideally we would like our method to provide only the output of the calculation to the parties involved, and reveal nothing more.  This is a lofty goal without the aid of trusted third parties, however it is relaxed in a useful way in the cryptographic literature.  First it is assumed that the parties are not able to quickly solve computationally hard problems (such as breaking RSA encryption).  Then, a protocol is secure so long as intermediate values in the computation either contain almost no information (in the sense that the protocol would have to be re-run astronomically many times on the same input data in order to detect any information in the messages), or will only reveal information as the result of an intractable computation.  We now briefly review the security model we intend to use.

We consider the ``functionality'' (see \cite{goldreich}) which maps the data of each party into the logistic regression parameter vector $\beta$:

\begin{equation}\label{lr-fun} \{(X_1,y_1),(X_2,y_2),\cdots (X_P,y_P)\} \rightarrow \{\beta, \beta, \cdots \beta \}
\end{equation}

The right hand side represents $P$ copies of the parameter, so that each party receives the same output.  Note that each design matrix is of the same dimensions.

A protocol for computing the functionality is just a sequence of steps, consisting of parties performing local computations, and sending intermediate messages to each other. In this work we build up a protocol for computing (\ref{lr-fun}) which is secure in the presence of ``semi-honest'' parties.  That is, parties who obey the protocol (and do not try to e.g., inject malformed data) but keep a transcript of all the messages they receive.  Intuitively, a protocol is secure in this setting whenever the intermediate messages give no information about the secret inputs of other parties.  Formally, the ``view'' of the $j^{th}$ party during the protocol is:

\begin{equation}\label{lr-view}
\text{view}_j((X_1,y_1),(X_2,y_2),\cdots (X_P,y_P)) = \{(X_j,y_j),r,m_1,\cdots m_{|m|} \}
\end{equation}
where $r$ is a record of all the random draws made by party $j$, and $m_k$ is the $k^{th}$ message received by that party (we have dropped dependence of $m$ on $j$ for readability).

The protocol is secure so long as there exists a polynomial time algorithm which, when given only the input and output of party $j$, may output a random transcript of message which is \emph{computationally indistinguishable} from $\text{view}_j$.   See  \citet{goldreich} for a definition and discussion of computational indistinguishability.  In essence, if the distribution of the sequence of messages depends only on the private input and output of party $j$ then we can simulate messages by drawing from this distribution (so long as the random number generator returns samples which are computationally indistinguishable from draws from the distribution).  The existence of a simulator shows that intermediate messages do not depend on the private input of other parties, and so the protocol is secure in the sense that parties gain no more information about each other's private inputs than that revealed by the output of the protocol.  Note that this type of security is ``orthogonal'' to that studied in \citet{Chaudhuri}, which seeks to prevent leakage of secret information in the parameter vector.

An example of a protocol which does not achieve this definition of security is one where all parties send their data to party 1, who computes the parameter locally on the combined data and then sends it back to all other parties.  In this case the messages received by party 1 consist of the data of other parties, in general it is impossible to simulate these messages given only the input and output belonging to party 1.

In the next section we present a protocol for performing Newton's method on the logistic regression objective in a way which is secure in the presence of semi-honest parties. Our protocol makes use of a specially designed approximation for the logistic function. Section \ref{sec:protocol2} then describes a different approximation necessitating  the operations of only sums and products, and thus speeding-up the computations.

Although we propose to use the cryptographic model for security, others exist and deserve a place in the theory of privacy preserving data analysis.  The main alternatives we see are ``weak'' security, and {\em perturbation} of the data.  The former comprises a body of literature summarized in \citet{ppdm_book}.  The idea is that by giving weaker privacy guarantees, we can implement much more efficient protocols.  Whether it is acceptable to have this weaker privacy guarantee is a question which one must   consider  on a case-by-case basis.  Although we describe our protocol in terms of the cryptographic model, by replacing the primitive operations (in Section~\ref{sec:blocks}) with their weakly-secure counterparts, we convert our protocol into a weakly secure (but also computationally more efficient) one.

The second alternative is data perturbation or sanitization.  The idea would be for each party to somehow perturb his data until he is happy to release it to the other parties (e.g., through the addition of random noise).  Thereupon the parties would each have a noisy copy of all the data, and could locally compute whatever statistical method they wanted on the union of the data.  The difficulty with this approach is that to protect privacy may require the addition of noise of such amplitude as to render the data itself useless.

\section{Primitives for Secure Protocols}\label{sec:blocks}

In this section we lay out some primitives and  sub-protocols which we will commbine to make a full logistic regression protocol.  While, details of the implementation of these primitives are  in the references cited, we also include some  in the appendix.

\subsection{Secret Sharing}
In our construction we  make extensive use of additive secret sharing.  The idea is to divide a quantity of interest $a$ into $P$ random numbers $a_j$ (one for each party) so that $\sum_ja_j = a$.  If the $a_j$ are distributed uniformly in the field (e.g., the entirety of $\mathbb{R}$) then any subset of the $a_j$ will reveal nothing about $a$.  In fact the sum over any subset is a random variable, the distribution of which does not depend on the secret value.

We use this construction to keep all intermediate quantities secret during the evaluation of Newton's method (i.e., the gradient, Hessian and intermediate parameter vectors).  As long as we can construct sub-protocols which compute random shares of a quantity, from random shares of inputs, then we can compose these sub-protocols together to finally obtain random shares of the logistic regression estimate.  With these in hand the parties can then exchange shares and reveal the vector itself.

Although the joint distribution of the $a_j$ concentrates on the linear subspace corresponding to the secret value, marginally the shares are uniformly distributed and do not  depend on any parameters.  Hence we can easily simulate messages based on these shares since the marginal distributions are known, and  we achieve security as defined in Section~\ref{sec:smpc} .  Next we show how to compute additive shares of all the intermediate quantities using the abstract definition of additive shares.  Although this approach is intuitively appealing, computers would quickly run into problems representing samples drawn uniformly from $\mathbb{R}$.  Therefore, in the appendix we show how to approximate arbitrarily well the same computations in modular arithmetic on $\mathbb{Z}_B = \{0,1,\cdots B-1\}$ for some large $B$.

\subsection{Computing Sums and Products with Random Shares}\label{sec_la}

To implement Newton's method we must essentially  perform linear algebraic operations on random shares, for example by computing shares of the Newton step from shares of the gradient and inverse Hessian.  In this section we describe how to obtain random shares of sums and products of quantities that are themselves represented as random shares.  Using these constructions, we   compute inner and outer products of vectors of random shares, and hence also matrix multiplies.

Computing shares of the sum of two secret quantities $a=\sum_j a_j$ and $b=\sum_j b_j$ is direct, as it involves only the local computation $a_j+b_j$ for each party $j=1,\ldots, P$. That is, party $j$ simply adds his shares $a_j$ and $b_j$ together to get a random share of the quantity $a+b$.
Obtaining random shares of the product of two secret quantities is more involved:

\begin{equation*} ab = \sum_ja_j\sum_kb_k = \sum_j{a_jb_j} + \sum_j\sum_{k\neq j}{a_jb_k} \end{equation*}

\noindent  The elements of the first sum on the right hand side can be computed locally by each party.  The second (double) sum, however,  requires products between random shares held by different parties.  To obtain these terms while maintaining the security of the protocol, we turn to \emph{oblivious function evaluation}.  That is, we pose the problem of computing the product as evaluating a function so that one party only knows the function and the other party only knows his input and the value of the function applied to that input.

The function set up by party $j$ and evaluated by party $k$ on his input, $b_k$, is:
\begin{equation}\label{prod-share-fn}
f_{(j \rightarrow k)}(x\, ; a_j)  = a_jx + r_{j,k} \; ,
\end{equation}

\noindent where $r_{j,k}$ is a quantity generated uniformly at random by party $j$.  Evaluation is done in a manner so that party $j$ learns nothing of the output (and thus only learns about $r_{j,k}$ which he generated) and party $k$ learns only the output.  Since party $k$ does not know the value of the random variable $r_{j,k}$ he has learned potentially nothing about the true value of the product.  Taking $m_{j,k} = -r_{j,k}$ and $n_{j,k} = f_{(j\rightarrow k)}(b_k\,;a_j)$ we have that $m_{j,k} + n_{j,k} = a_jb_k$, and thus they form random shares of the product $a_jb_k$.

Once parties compute random shares  for all the terms $a_jb_k$, they can locally compute random shares of $c=ab$ as:
\begin{equation}\label{eq:prod_shares}
c_j = a_jb_j + \sum_{k\neq j}(n_{k,j} + m_{j,k}) = a_jb_j + \sum_{k\neq j}(f_{(k \rightarrow j)}(a_j\,;b_k) - r_{j,k})  \; .
\end{equation}
Summing up these quantities, and utilizing the definition of the linear function set up in (\ref{prod-share-fn}), we easily obtain:
\begin{equation}
\sum_j c_j = \sum_j\big\{a_jb_j + \sum_{k\neq j}(a_jb_k + r_{k,j} - r_{j,k})\big\} = \sum_j{a_j\sum_k{b_k}} \; , \nonumber
\end{equation}
which shows that the $c_j$'s in (\ref{eq:prod_shares}) are indeed (additive) shares of the product.

This protocol generates random shares of the product even if the original shares weren't themselves random, e.g., if they were due to the partitioning of the data.  A method which implements these encrypted multiplications using fixed-point arithmetic is given in \cite{fhn:10}.  

We also note that dividing one secret value into another securely is much more difficult than dealing with products and requires more elaborate (and computationally demanding) protocols.  Below we show how matrix inversion can be performed without any divisions.

\subsection{Evaluating Interval Membership}\label{sec:yaosgt}
We suppose we are able to evaluate the following predicate in a secure way:
$$1\{a \geq b\}$$
Where $a,b$ are secret values held by separate parties.  This is known as Yao's ``millionaires problem,'' since he described it in the context of determining which millionaire has the most money, without disclosing actual bank balances.

An example of a protocol which computes this predicate is given by \citet{gt_proto}.  We can also trivially extend it so that each party receives a random share of the output bit (i.e., each party receives a random bit, the ``xor'' of which yields the correct output bit).  Using this technique we can also check whether a secret value (i.e., a sum of random shares) is greater or less than some constant:

\begin{equation}\label{yaos-gt-eqn} 1\{a_1+a_2 \geq c\} = 1\{a_1 \geq c-a_2\},\end{equation}

\noindent where $a_1,a_2$ are the random shares of $a$ held by two parties.

\subsection{Securely Inverting a Matrix}\label{sec:matrix_inversion}

We use a matrix inversion routine  built up entirely of matrix multiplications and subtractions, thus allowing us to use the constructions of the preceding sections to implement it securely. We obtain the reciprocal of a number $a$  without necessitating any actual division  by an application of Newton's method to the function $f(x)=x^{-1}-a$. Iterations follow $x_{s+1}=x_s(2-ax_s)$, which requires multiplication and subtraction only.

It turns out that we can apply the same scheme   to matrix inversion,  e.g.,  see~\cite{guo_higham} and references therein. A numerically stable, coupled iteration for computing $A^{-1}$, takes the form:
\begin{equation}\label{matrix-inv}
\begin{array}{llllll}
X_{s+1} &=& 2X_s-X_sM_s  & \quad , \quad X_0 &=& c^{-1}I \; ,\\
M_{s+1}&=& 2M_s-M_s^2 & \quad , \quad  M_0 &=& \; c^{-1}A,
\end{array}
\end{equation}
where $M_s = X_s A$, and $c$ is to be chosen by the user. A possible choice, leading to a quadratic convergence of $X_s\rightarrow A^{-1}$ $(M_s \rightarrow I)$, is $c=\max_i \lambda_i(A)$. In our actual implementation we used instead the trace (which dominates the largest eigenvalue, as the matrix in question is positive definite), since we can compute shares of the trace from shares of the matrix locally by each party. To compute $c^{-1}$ we use the same iteration, with scalars instead of matrices.  For this iteration we initialize with an arbitrarily small $\epsilon>0$ (as convergence depends on the magnitude of the initial value being lower than that of the inverse we compute).  We use the constructions of section~\ref{sec_la} to iterate through (\ref{matrix-inv}) until convergence. As $M_s\rightarrow I$, we check for convergence by considering the absolute difference between the trace of $M_s$ and the data dimension $d$, and we can evaluate the function $1\{|\text{tr}(M_s)-d|>\epsilon\}$ on random shares of the trace of $M_s$ using the same form as (\ref{yaos-gt-eqn}).

\section{First Protocol for Logistic Regression}\label{sec:protocol1}

We recall the usual Newton-Raphson iteration expression (\ref{eq:nr}). 
To perform the iteration  we first  compute random shares of the update direction: $\Delta_t = -(\nabla^2\ell(\beta_t))^{-1} \nabla\ell (\beta_t)$,
\noindent via the formulation of matrix-vector products of random shares.  
We can then add these random shares to the current parameters $\beta_t$ to obtain random shares of $\beta_{t+1}$.  To check convergence recall (from e.g., \cite{cvx_book}) we should end if:
\begin{equation}\label{convergence}
\lambda^2 = (\nabla\ell(\beta_t))^T\Delta_t \leq \epsilon \; .
\end{equation}
We can compute (\ref{convergence}) securely using the same form as (\ref{yaos-gt-eqn}).  The result is sharable among all the parties, and the protocol ends whenever the result is 0, i.e., when $\lambda^2$ is not greater than $\epsilon$.

By using the constructions of the previous section, we have the tools required to invert shares of the Hessian, and thus to compute the Newton step.  All that we need to do is  construct a secure protocol to evaluate the logistic (sigmoid) function.  In principle, a specialized sub-protocol could be built up using the construction of Yao \cite{yao82}.  The method would be to construct circuit that evaluates the sigmoid function in the same manner that the arithmetic logic unit in a CPU would.  Then we could give this circuit the secure treatment and make it into a protocol following \citet{goldreich}.  The disadvantage with this approach is that the circuit evaluation protocols are prohibitively expensive and thus they are not  useful in practice except for trivial circuits, see e.g.,  \citet{fairplay}.   Instead we use a specially crafted approximation to the logistic function in terms of indicator functions.  We describe this next.

\subsection{A Secure Approximation to the Logistic Function}\label{sec:RS_of_sigma}

The logistic function itself is the CDF of the logistic distribution.  We propose to approximate this function with an ``empirical CDF.''  This is a function of a set of $L$ samples $z_l$, taken independently from a logistic distribution:

\begin{equation}\label{sigma-appox}
\sigma(a) \approx F_L(a) = L^{-1}\sum_{l=1}^L{1\{a \geq z_l\}} \; .
\end{equation}

Based on the Glivenko-Cantelli theorem, and later work by Dvoretzky, Kiefer and Wolfowitz, the rate at which the empirical CDF converges to the true CDF (i.e., the logistic function which is of interest) is known.  Using these results, we  obtain bounds on the maximum difference between the logistic function and our approximation, which hold with high probability.  See the remark below in Section~\ref{sec:ECDF} about the accuracy of the approximation.

We now turn attention to obtaining random shares of the logistic function evaluated at random shares of $\beta^Tx_i$.  We obtain random shares of $\beta^Tx_i$ by using the inner product construction for multiplying together random shares.  If we denote shares of this inner product by $(\beta^Tx_i)_j$ for party $j$, we  write:
\begin{equation} \label{eq:sigma-approx2}
\sigma(\beta^Tx_i) \approx L^{-1}\sum_l{1\{\beta^Tx_i \geq z_l\}} = L^{-1}\sum_l{1\{(\beta^Tx_i)_1 + (\beta^Tx_i)_2 \geq z_l\}} \; .
\end{equation}
Thus the problem reduces to getting random shares of the sum of indicators.   Note that we can re-write each indicator function as:
\begin{equation} \label{gt-shares}
1\{(\beta^Tx_i)_1 + (\beta^Tx_i)_2 \geq z_l\} = 1\{(\beta^Tx_i)_1 \geq z_l - (\beta^Tx_i)_2 \} \; .
\end{equation}

\noindent If party 2 generates the logistic random variables then we have a trivial reduction to (\ref{yaos-gt-eqn}).  In order to restrict the view of either party to a random share, we restrict the output to random bits $o^l_1$, and $o^l_2$, such that
\[
o^l_1\oplus o^l_2 = \left\{
\begin{array}{cc}
1 & \mbox{if} \quad 1\{a\geq z_l\}   \\
0 & \mbox{otherwise}
\end{array}\right. \; ,
\]
where $\oplus$ is the exclusive or. The right-hand side of equation (\ref{eq:sigma-approx2}) requires (random shares of) the fraction of outputs with $o^l_1\oplus o^l_2=1$. This can be established by noticing that
\[
\sum_{l=1}^L (o^l_1\oplus o^l_2)=\sum_{l=1}^L o^l_1+\sum_{l=1}^L o^l_2-2o_1^To_2 \; ,
\]
where we denote $o_k=(o^1_k, \ldots, o^L_k)$ for $k=1,2$. 
Jagannathan and Wright~\cite{jagannathan} use this method to convert xor shares into additive shares for a different privacy-preserving task.

In order for the output to behave this way, we can either use Yao's protocol directly, or take a (more efficient) GT protocol and modify it to give a (xor) random share.  In this work we use the protocol of \citet{gt_proto}.  
Having computed random shares of the logistic function, we can use the constructions of Section~\ref{sec_la} to compute random shares of the gradient and Hessian, and hence build a full logistic regression protocol.

\subsection{Quality of the Logistic Approximation}\label{sec:ECDF}
A comment about the accuracy of approximation (\ref{sigma-appox}), and the resulting logistic parameter estimator is in order. The tail behavior of the sup-norm of the error is given, for every $\epsilon>0$, by:
\begin{equation}\label{dkw_ineq}
P \Big(\big\|\sigma(\cdot)-L^{-1}\sum_{l=1}^L{1\{\cdot \geq z_l\}}\big\|_\infty>\epsilon \Big)\leq  2e^{-2L\epsilon^2} \; ,
\end{equation}
known as the Dvoretzky-Kiefer-Wolfowitz (DKW) inequality. One possible way to choose the number of Logistic variables $L$ in practice, is by ensuring that the above probability is no more than a prescribed level of accuracy, say $\alpha$. Solving for $L$ we obtain the (very conservative) bound $L\geq -\frac{1}{2}\epsilon^{-2} \log(\alpha/2)$. A less conservative bound might entail the maximum absolute error restricted to some interval (containing the origin). 

We relate error in the approximation of the sigmoid (and hence the gradient) to the error in the convergent parameter by the following inequality:

\begin{equation}\label{parm_err_bound}
||\hat\beta-\beta||_2 \leq \frac{R [L^{-1}+\|\sigma(\cdot)-F_L(\cdot)\|_\infty ] }{\hat\lambda_{\text{min}}} \; ,
\end{equation}
in which $\hat\beta$ is the optimizer of the exact log likelihood, and $\beta$ is the optimizer of our approximation, $\hat\lambda_{\text{min}}$ is the smallest eigenvalue of the Fisher information matrix $I(\beta) = -n^{-1}\nabla^2\ell(\beta)$ (on some interval, see Appendix~\ref{sec:validity_1}), and $R$ is the radius of a ball which containing all the data vectors (i.e., $\forall i, ||x_i||_2 \leq R$).  The proof of this inequality follows lemma 1 of \cite{Chaudhuri} in which the two convex functions are the exact log likelihood objective, and the difference between the exact and approximate objectives. See Appendix~\ref{sec:validity_1} for detailed theoretical derivation.

Since we can use expression (\ref{dkw_ineq})   to bound the numerator of  expression (\ref{parm_err_bound}),   the parameter output by our protocol, and that output by the exact (non-private) algorithm can be brought arbitrarily close (except on a set of negligible probability) by increasing the parameter $L$.  Later, we perform an experiment to show how well the method performs with reasonably small $L$.  Note that for Newton's method to converge in this approximation, we must use the same sample of $L$ logistic random variables each time we approximate the sigmoid.  Otherwise assessing convergence would be difficult as the objective function would be constantly shifting.  We propose that the parties draw $L$ logistic variables ahead of time, and use these for all the computations.

\subsection{Hessian Lower Bound Technique}\label{sec:hess_bound}

Notice that the Newton Raphson method requires inverting a matrix (the Hessian of the log likelihood) at each iteration.  In our setting, using our iterative inversion method this becomes very expensive.  Therefore we propose to use a well-studied approximation \cite{minka}, which replaces the iteration by:

\begin{equation}\label{eq_hessian_lb} \beta_{t+1}=\beta_t - 4(X^TX)^{-1} \nabla \ell \; .
\end{equation}

First note that under this technique the algorithm only ever needs a single matrix inversion, since $X^TX$ is constant throughout all the iterations.  Second, this algorithm still eventually converges to the correct parameter value (modulo the other approximations we make in our protocol).  The reason is that the inverse hessian is always greater than $4(X^TX)^{-1}$, in the sense that the difference is positive semi-definite, see e.g, Minka \cite{minka} for more details.  What's more, this technique ensures that progress towards the optimum is monotonic,  and so  assessing convergence may be simpler.

\subsection{Computation and Communication Complexity}\label{sec:complexity_1}

First we count how many times we must run each of our primitives for each iteration of Newton's method.  
The approximation of section~\ref{sec:RS_of_sigma} requires $nL$ instances of the GT protocol per round, as $L$ instances are required per case.  Computing the gradient 
and the Hessian requires $n(1+d+d^2)$ multiplications.
Inverting the Hessian takes $2d^3$ multiplies and one GT per iteration of (\ref{matrix-inv}).  Since this inner iteration is quadratically convergent, it takes $O(\log{d})$ iterations to converge, and thus takes $O(d^3\log{d})$ multiplies and $O(\log{d})$ instances of GT. 
 In total then, each outer iteration takes $O(nd^2+d^3\log{d})$ multiplies, and $nL+O(\log{d})$ invocations of the GT protocol.

Each multiplication requires a number of encryptions and decryptions;  this scales quadratically with the number of parties $P$ since they must exchange with one another.  Thus the computational workload increases as the data are split into more pieces.  Note that although repeated use of the cryptosystem is quite expensive, performance  on normal hardware is relatively rapid.  A machine  dedicated to the computation and running multiple threads can do thousands of encryptions per  second.

Each instance of GT using the protocol of \cite{gt_proto} requires $O(\log{B})$ encryptions and decryptions (and operations on encrypted values etc.).  Therefore in total our approximation of section~\ref{sec:RS_of_sigma} requires $O(nL\log{B})$ encryptions per iteration.  This may be too computationally demanding for large $L$.  One way to reduce this cost is to run the scheme using a coarse approximation to the sigmoid (i.e., a small $L$) to convergence, then increase $L$, resample the logistic variables and then continue Newton's method from the previous convergent parameter.  Although the latter iterations will still be computationally burdensome, there will be fewer of them. Another way is to use a different approximation to the sigmoid function. This is outlined next in Section \ref{sec:protocol2}.

Note that the total amount of communication by all parties is also proportional to the number of multiples and GT invocations.  For an invocation of either, a party must transmit $\log{N}$ bits to another party, and then receive a message of the same length. There are a total of $O(P^2(nd^2+d^3\log{d})+nL)$ messages which must be sent for each iteration.  If  the number of parties or cases, or the granularity of the approximation is large,   running the protocol over a high speed local area network would make the communication overhead manageable.

\section{Second Protocol for Logistic Regression}\label{sec:protocol2}

As we mentioned above, the computation complexity of evaluating  approximation (\ref{sigma-appox}) to the logistic function scales linearly with $L$, since  on each of Newton's iteration we invoke Yao's protocol to compute the GT predicate, and we do it for every case $i$.    This may be prohibitively expensive even for a moderate $L$. A possible way to reduce this computational burden was briefly described in Section \ref{sec:complexity_1}. Here, we provide full details of a more structured approach, which is reminiscent of Euler's method. The approach is built (again) on computing Newton's iteration (\ref{eq:nr}). It would be more natural in this section to treat the logistic function in a {\it vectorized} fashion, i.e., $\sigma(a)=(\sigma(a_1),\ldots, \sigma(a_n))$, for an $n$-dimensional vector $a=(a_1, \ldots, a_n)$. Therefore, we use different, albeit equivalent, representations for the gradient and Hessian:
\begin{equation}
\nabla \ell = X^T\{y-\sigma(X\beta)\} \quad , \quad \nabla^2\ell=-X^T \text{diag}\{\sigma(X\beta)\circ (1-\sigma(X\beta))\}X \; .
\end{equation}
Here $X$ is the design matrix whose rows are $x_i^T$, the units or feature vectors (see (\ref{grad-hess})). The symbol ``$\circ$'' denotes the element-wise product, i.e., $u\circ v=v\circ u=\text{diag}(u)v$.


We modify the iteration so that we neither  explicitly compute the logistic function $\sigma(\cdot)$ which is involved in both the gradient and Hessian, nor use the approximation in expression (\ref{sigma-appox}).  Note that throughout the procedure we may treat each unit $x_i$ as having an associated logistic function value $\sigma(\beta_t^T x_i)$.  We propose to track a vector of approximate function values $\hat{\sigma}_t \approx \sigma(X\beta_t)$ which will be updated after each iteration.  Then, these approximate values will be used to compute the next iteration of $\beta_t$.  Note that the derivative of the logistic function is given by:
\begin{equation}\label{logistic_deriv} \sigma^\prime(a) = \sigma(a)(1-\sigma(a)) \stackrel{\text{def}}{=} g(\sigma(a)) \; .
\end{equation}
Therefore, knowing the value $\sigma(a)$, we can determine the derivative of the logistic function around $a$   by a single multiplication.  Linearizing around some value $a_0$ gives:
\begin{equation}\label{sigma_lin} \sigma(a) = \sigma(a_0) + (a-a_0)g(\sigma(a_0)) + 2^{-1}(a-a_0)^2\sigma^{\prime\prime}(\cdot)\big|_{a^\star} \approx \sigma(a_0) + (a-a_0)g(\sigma(a_0)) \; ,
\end{equation}
where the second derivative is evaluated at some value $a^\star$ in the interval between $a$ and $a_0$.  Denote by $\Delta_t = \beta_{t+1}-\beta_t$ as in Section \ref{sec:protocol1}, then may make use of the approximation:
\begin{equation}\label{sigma_approx} \hat{\sigma}_{t+1} = \hat{\sigma}_t + (X\Delta_t)\circ g(\hat{\sigma}_t) \; ,
\end{equation}
where $g$ is applied element-wise to $\hat{\sigma}_t$.

Over the course of the entire algorithm, the approximation $\hat{\sigma}_t$ is updated repeatedly, in a manner very similar to using Euler's method to numerically integrate the differential equation (\ref{logistic_deriv}).  It is well known that the error of this method decreases with the size of the ``step'' taken at each iteration. In the above, the steps are of size $X\Delta_t$, which will in general be different on each iteration, and will also be different for each unit.  In order to control the error we amend this approximation by breaking down the step into $k$ smaller steps each of size $k^{-1}X\Delta_t$, and performing $k$ such updates.  As we shall see, we may base our choice of $k$ on some aspect of the design matrix,  $X$, in order to reach a desired level of error in the approximation.  We write this approximation as:

\begin{equation}\label{sigma_approx_k}
\hat{\sigma}_{t+1} = \hat{\sigma}_t + k^{-1}X\Delta_t \circ \sum_{i=1}^k g(\hat{\sigma}_i^\star) \stackrel{\text{def}}{=} \hat{\sigma}_t + X\Delta_t \circ \tilde{g}_k(\hat{\sigma}_t, X\Delta_t) \; ,
\end{equation}
where the $\hat{\sigma}_i^\star$ are the intermediate values corresponding to the inner iterations, and we define $\tilde{g}_k$ as the function which gives the average value of $g$ evaluated on these values.

We summarize our method  in the following coupled iteration:
\begin{eqnarray}
\beta_0 &=& 0^{d\times 1} \nonumber\\\nonumber
\hat{\sigma}_{0} &=& 2^{-1}\cdot 1^{n\times 1} \\
\Delta_t &=& 4(X^TX)^{-1}X^T(y-\hat{\sigma}_{t}) \label{eq:coupled}  \\  \nonumber
\beta_{t+1} &=& \beta_t + \Delta_t \\ \nonumber
\hat{\sigma}_{t+1} &=& \hat{\sigma}_{t} + X\Delta_t \circ \tilde{g}_k(\hat{\sigma}_t, X\Delta_t) \; , \nonumber
\end{eqnarray}
where $0^{d\times 1}$ is the $d$-dimensional vector of zeros and $1^{n\times 1}$ is the $n$-dimensional vector of ones.
The proposed iteration differs from the protocol of Section \ref{sec:protocol1} (and from the usual Newton-Raphson method).  The main difference is that we have replaced the logistic function approximation (\ref{sigma-appox}) with our Taylor approximation. Note that we are using again the bound on the Hessian (see Section \ref{sec:hess_bound}), which would make computation easier. We use this technique in our method for this reason, and also since it interacts well with our Taylor approximation by ensuring that convergence towards the optimum is in a sense monotonic, as shown in Section \ref{sec:convergence}.  In keeping with our goal of using only sums and products, we recall that it is possible to invert a matrix with just these operations (see Section \ref{sec:matrix_inversion}).


We now present a bound on the distance from our approximated regression coefficients $\beta_t$, to the true optimizer of the log-likelihood which we denote by $\hat\beta$, as in (\ref{parm_err_bound}).
Since our iterations are guaranteed to converge (see Section \ref{sec:convergence}), we can choose to run the iterations until $||X^T(y-\sigma_t)||_2$ is smaller than some threshold $b$ (i.e., by choosing $t$ accordingly):
\[ b  \geq ||X^T(y-\hat{\sigma_t})||_2  \geq ||X^T(y-\sigma_t)||_2 - ||X^T(\hat{\sigma}_t - \sigma_t)||_2  \; .
\]

Therefore we can bound the norm of the gradient of the logistic log-likelihood taken at our final parameter estimate:
$$||\nabla\ell(\beta_t)||_2^2 \leq b + nRc\tau \; .$$
where $R$ is the radius of a ball containing all the data vectors, exactly as in (\ref{parm_err_bound}), $c$ is some constant, and $\tau$ is a quantity upper bounding the maximal Euler's step size.

We can use this to construct our main result about the quality of our approximation.  Suppose we choose $b\leq nRc\tau$, then from the above we have:
\begin{equation}\label{eq:error_2}
||\beta_t-\hat\beta||_2 \leq \frac{2Rc\tau}{\hat{\lambda}_{\text{min}}} \; ,
\end{equation}
where $\hat{\lambda}_{\text{min}}$ is the smallest eigenvalue of the Fisher information matrix $I(\cdot) = -n^{-1}\nabla^2\ell(\cdot)$ in the line segment between $\beta$ and $\hat\beta$.  Note that $\hat{\lambda}_{\text{min}} = n^{-1}\lambda_{\text{min}}$ and  the factors of $n$ cancel.

Therefore we can make the accuracy of our approximation arbitrarily good by decreasing $\tau$ although, as we shall see there is a tradeoff involved.  A smaller $\tau$ usually means a higher $k$, resulting in increased computational demands. We refer the reader to Appendix~\ref{sec:validity_2} for complete technical details.

\subsection{Choice of $k$}
Thus far we have that the error of the approximation decreases as $\tau$ is decreased; however, this last variable is not controlled directly (as $L$ was in protocol 1) but rather is a function of $k$, the number of steps taken for each outer iteration of the algorithm.

In principle, to get at a prescribed step size $\tau$, we can choose $k$ by noting that:

\begin{equation}\label{eq_choose_k}\tau \leq ||k^{-1}X\Delta_t||_\infty \leq k^{-1} ||X(X^TX)^{-1}X(y-\hat{\sigma}_t)||_2 \leq k^{-1}\sqrt{n}\end{equation}

This leads to the overly conservative choice of $k=\tau^{-1}\sqrt{n}$.  An alternative choice is to run the protocol with a small value of $k$, e.g., 10, and then to re-run with different values to assess the sensitivity of the computation.  In Section~\ref{sec:experiment} we show that this technique performs wekk even with small $k$.

\subsection{Computational Complexity}

We can measure the overall complexity of our method in terms of the number of products that are needed, since these are the most time-consuming operations we use.  First note that to construct the matrix $X^TX$ takes $nd^2$ products, and inversion of this matrix takes $O(d^3)$ using (\ref{matrix-inv}), where the constant is related to the condition of the matrix.  Then on each iteration, to compute $\Delta_t$ takes $nd + d^2$ products.  Our approximation to the logistic function takes $nk$ products, for a total of $n(k+d) + d^2$ products per iteration.

We compare this with the cost of a protocol which computes the logistic function via a specially designed sub protocol based on circuit evaluation, cf. \citet{yao82}.  If the latter may be evalutated using $q$ encryptions, then the complexity would be $n(d+q) + d^2$ operations per iteration.  As mentioned before this number would typically be much larger than $k$ (for example on the order of the number of bits used to represent the numbers).  Therefore on each iteration we can save a multiple of $n$ operations, which may be especially important when $n$ is large.

\section{Security Guarantees}\label{sec:security}

Since our protocol runs until convergence, the number of rounds is variable and depends on the data itself.  Furthermore a matrix inversion was performed by an iterative scheme which itself took some variable number of iterations to converge.  Therefore we amend the protocol so that the output for each party is a triple consisting of the convergent parameter value, and the number of iterations it took to converge, and the number of iterations taken for the matrix inversion.  This way the messages received from testing convergence are easily simulated (i.e. a zero on every round up until the number specified in the output, then a one on that iteration) and this clearly reveals no more information since the parties know ``where they are'' in the protocol at all times and  could count these numbers of iterations.  Having dealt with this technicality we will consider simulating the other intermediate messages in our simulator, and consider these convergence tests already taken care of.

In both of our protocols, the messages which are transmitted are always part of some sub-protocol, namely multiplication or evaluation of the ``greater than'' predicate.  The only exception to this is the final messages which are sent immediately before the output is reconstructed.  As those messages are themselves random shares they may be simulated easily (although they must be simulated in such a way that they sum to the correct output values, but this is trivial).  The messages which are passed during the sub-protocols may be simulated based on their respective input and outputs so long as the sub-protocols are cryptographically secure.  Since we take care to ensure that the intermediate values are random shares, the simulators for the sub-protocols ``compose'' to form a simulator for the main protocol (see \cite{goldreich}).

\section{Illustrative Experiments}\label{sec:experiment}

We provide two illustrative experiments to demonstrate our approach. The first aims at  showing the performance of our protocol from Section \ref{sec:protocol1}. Specifically, we examine the effect of  approximation (\ref{sigma-appox}) on the resulting parameter values, when small, and large number of logistic variables $L$ are being used. The second example takes a look at the altered protocol from Section \ref{sec:protocol2}, which uses the coupled iteration (\ref{eq:coupled}) instead of approximation (\ref{sigma-appox}), and reports its performances for different values of $k$, the number of Euler's ``steps''.

For both experiments we use an extract from the Current Population Survey (CPS) data (see \url{http://www.bls.gov/cps/}), which includes data on a sample of slightly more than  50,000 U.S. households.  We focus on predicting whether  household income is greater than 50,000 dollars.  We converted $M$-category features into $M-1$ binary features, and divided age into 4 bins corresponding to 20 year intervals.  Note that although we expressed our approach in terms of continuous covariates, it handles binary flags just as well, where said covariates take on e.g., $0.0$ and $1.0$.

Our protocol from Section \ref{sec:protocol1} deviates from the exact computation in two ways, first we use an approximation to the gradient, and second we perform all the calculations in fixed-point arithmetic.  Both of these approximations can be made arbitrarily tight but at the expense of computational efficiency.  To demonstrate that our protocol can be implemented in an efficient manner and produce reasonably accurate results we implemented it in a simulator and compared the results to exact logistic regression on the CPS data.


For each of $L=100$ and $L=500$ we ran our first protocol 100 times. The table below shows the means and standard deviations of the resulting parameter values.  Evidently, as $L$ gets bigger, the accuracy of the parameter values improves.  Figures \ref{fig_like1} and \ref{fig_like1a} show how the likelihood of the estimate maintained by the protocol increases with the number of iterations.  We computed the error bars   by removing the 5 samples that deviated from the mean by the greatest amount, and plotting the minimum and maximum from the remaining ones.  This then corresponds to an approximate 95\% confidence interval, and would become an exact interval if we were to perform more and more simulations.  For the purposes of comparison, we also plotted the likelihood achieved by the exact non-private Newton Raphson algorithm, and a non-private algorithm which we referred to as ``hessian lower bound.''  Both give upper bounds for what we  hope to achieve, the latter is an algorithm where we just use the approximation of (\ref{eq_hessian_lb}), and exact (i.e., non-private) logistic sigmoid values.  We see that as $L$ increases, the first protocol more closely approximates the hessian lower bound technique, which converges more slowly than the exact Newton Raphson method.

For the second experiment, we ran our coupled iteration on the CPS data with $k=5,10$. Although each iteration of our algorithm may be cheap, all is for nought if we require many more iterations for convergence. To determine whether this happens we compared our method to the Hessian lower bound method of (\ref{eq_hessian_lb}), since this represents our algorithm without the approximation. In figure \ref{fig_like2}, we plot the likelihoods of the second protocol against the iteration number.  Since there is no randomness in the second approximation, there are no error bars.  Even for small values of $k$, much smaller than those suggested by (\ref{eq_choose_k}), the approximation to the Hessian lower bound technique is quite good, and increasing $k$ further (e.g., to 50) results in curves which are exactly the same as that of the Hessian lower bound method.  In table \ref{tab:protocol2} we show the resulting parameter estimates for both methods.

\begin{table}
\begin{tabular}{|l | r | r r | r  r |  r | r | r|}
\hline

& NR & P1, $L=100$ & s.d. & P1, $L=500$ & s.d. & P2, $k=5$ & P2, $k=10$\\
\hline
Intercept & -10.7536 & -11.6306 & 1.0761 & -11.5732 & 0.5339 & -10.7944 & -11.2262 \\ \hline
Child Sup &0.0002 &  0.0002 & 0 & 0.0002 & 0 & 0.0002 & 0.0002 \\ \hline
Property Tax & 0.0003 &  0.0003 & 0 & 0.0003 & 0 & 0.0003 & 0.0003 \\ \hline
Num in Household & 0.9802  & 0.9916 & 0.0863 & 0.9881 & 0.0434 & 0.9259 & 0.9601 \\\hline
Num Children & -1.056  & -1.0721 & 0.0935 & -1.0685 & 0.047 & -1.0017 & -1.0384 \\\hline
Num Married & 0.0342  & 0.0343 & 0.0053 & 0.0343 & 0.0022 & 0.032 & 0.0333 \\\hline
Child Sup Ind. & -0.0001  & -0.0001 & 0 & -0.0001 & 0 & -0.0001 & -0.0001 \\\hline
Education & 0.3218 & 0.3276 & 0.0295 & 0.3265 & 0.0146 & 0.3058 & 0.3172 \\\hline
\multirow{4}{*}{Age} & -4.6178  & -3.9701 & 0.3451 & -3.94 & 0.1638 & -3.6202 & -3.8011 \\
& -4.2368  & -3.6782 & 0.3148 & -3.6596 & 0.1504 & -3.4817 & -3.5823 \\
& -3.9608 & -3.3355 & 0.2852 & -3.3157 & 0.135 & -3.1592 & -3.2481 \\
& -4.9575 & -4.4069 & 0.3795 & -4.3814 & 0.1829 & -4.1096 & -4.2624 \\ \hline
\multirow{6}{*}{Marital Status} &  0.0064 & 0.0051 & 0.0282 & 0.0001 & 0.012 & -0.0035 & -0.0007 \\
& -0.5362 & -0.5623 & 0.058 & -0.5562 & 0.0279 & -0.5205 & -0.5404 \\
& -0.4378 & -0.4718 & 0.0819 & -0.4609 & 0.0374 & -0.3934 & -0.4321 \\
& -0.5855 & -0.6096 & 0.0675 & -0.6018 & 0.0309 & -0.572 & -0.5889 \\
& -0.9617 & -1.0283 & 0.1146 & -1.0116 & 0.0549 & -0.957 & -0.9874 \\
& -0.7496& -0.7888 & 0.0852 & -0.7783 & 0.0401 & -0.736 & -0.7598 \\ \hline
-\multirow{3}{*}{Race} & 0.2417 & -0.2588 & 0.0276 & -0.2565 & 0.0135 & -0.2401 & -0.2491 \\
& -0.4981 & -0.5167 & 0.05 & -0.5128 & 0.0244 & -0.4835 & -0.4997 \\
& -0.1658 & -0.1796 & 0.0196 & -0.1785 & 0.0102 & -0.1631 & -0.1719 \\ \hline
Sex & -0.2763 & -0.2802 & 0.0244 & -0.2792 & 0.0125 & -0.2613 & -0.2712 \\
\hline
\end{tabular}
\caption{Estimates produced by the exact method (Newton Raphson), and the two protocols, for different parameter settings of the protocols.}
\label{tab:protocol2}
\end{table}

\begin{figure}[h]
  \centering
      \includegraphics[width=7in]{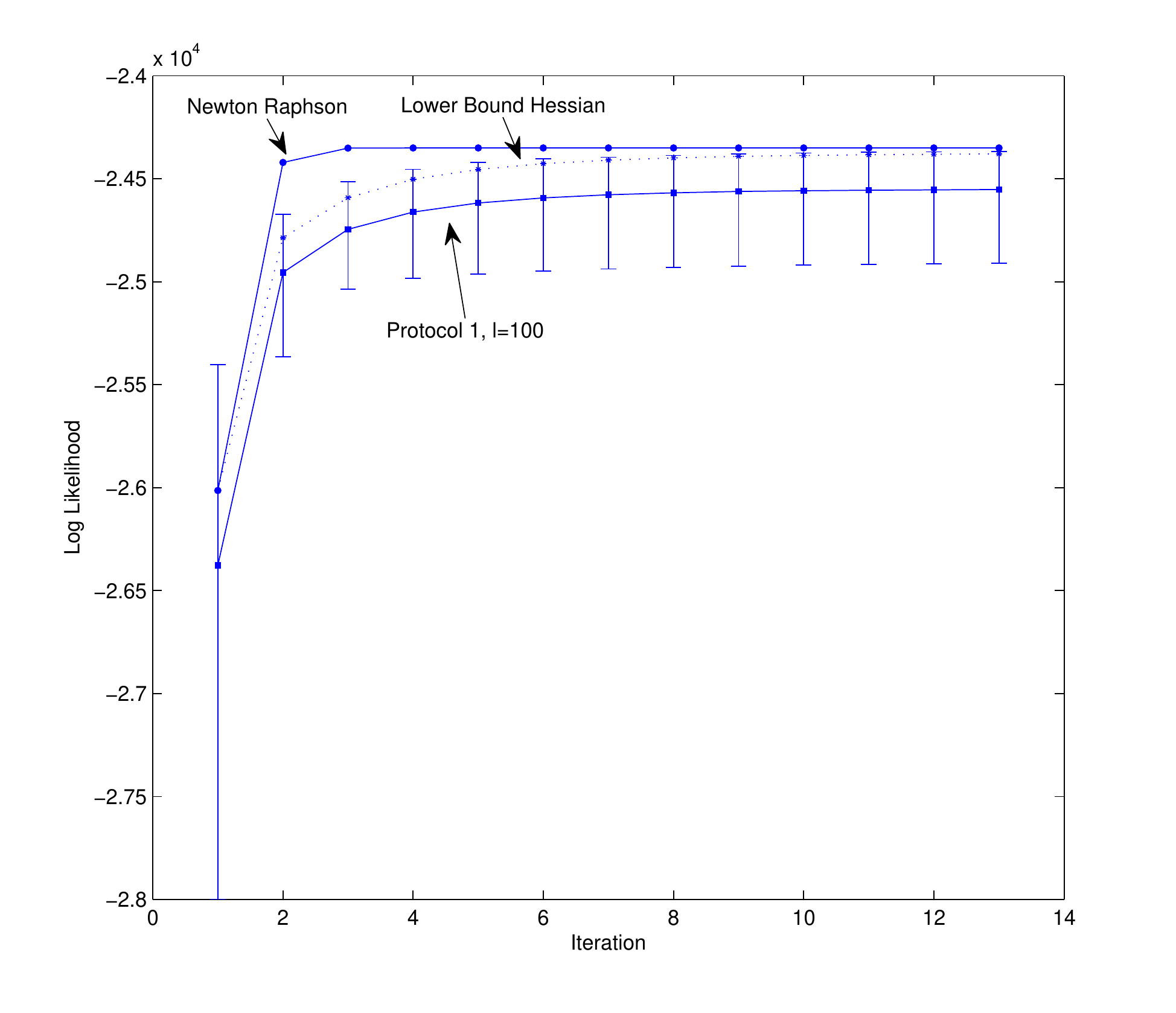}
  \caption{Log Likelihood vs iteration number for protocol 1 with $L=100$, and that of the ``Hessian Lower bound'' algorithm, which is the same as protocol 1 except with exact sigmoid evaluations.  We also compare to the full newton raphson method, which inverts the hessian on each iteration.}
\label{fig_like1}
\end{figure}

\begin{figure}[h]
  \centering
      \includegraphics[width=7in]{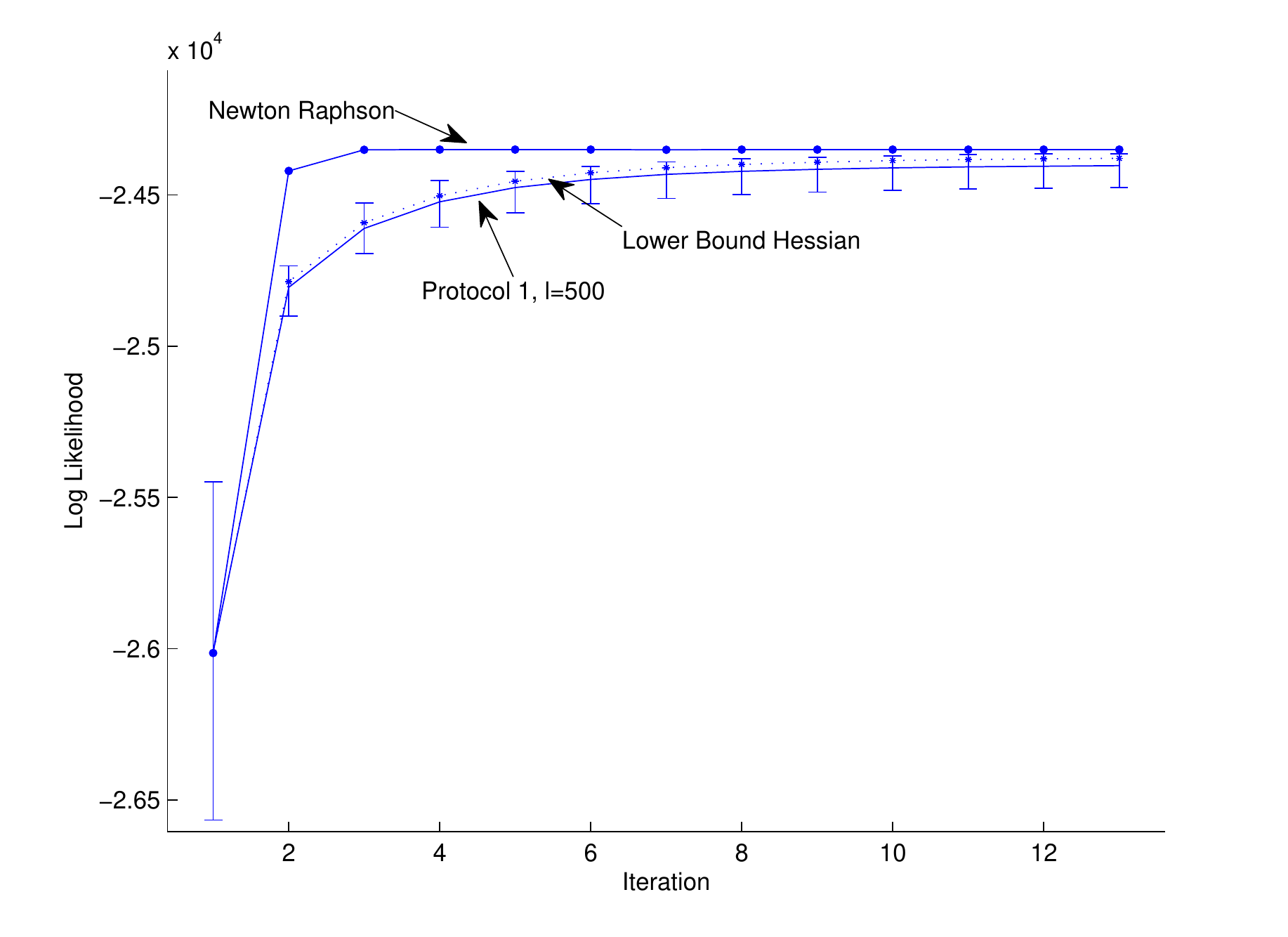}
  \caption{Log Likelihood vs iteration number for protocol 1 with $L=500$, and that of the ``Hessian Lower bound'' algorithm, which is the same as protocol 1 except with exact sigmoid evaluations.  We also compare to the full newton raphson method, which inverts the hessian on each iteration.}
\label{fig_like1a}
\end{figure}

\begin{figure}[h]
  \centering
      \includegraphics[width=7in]{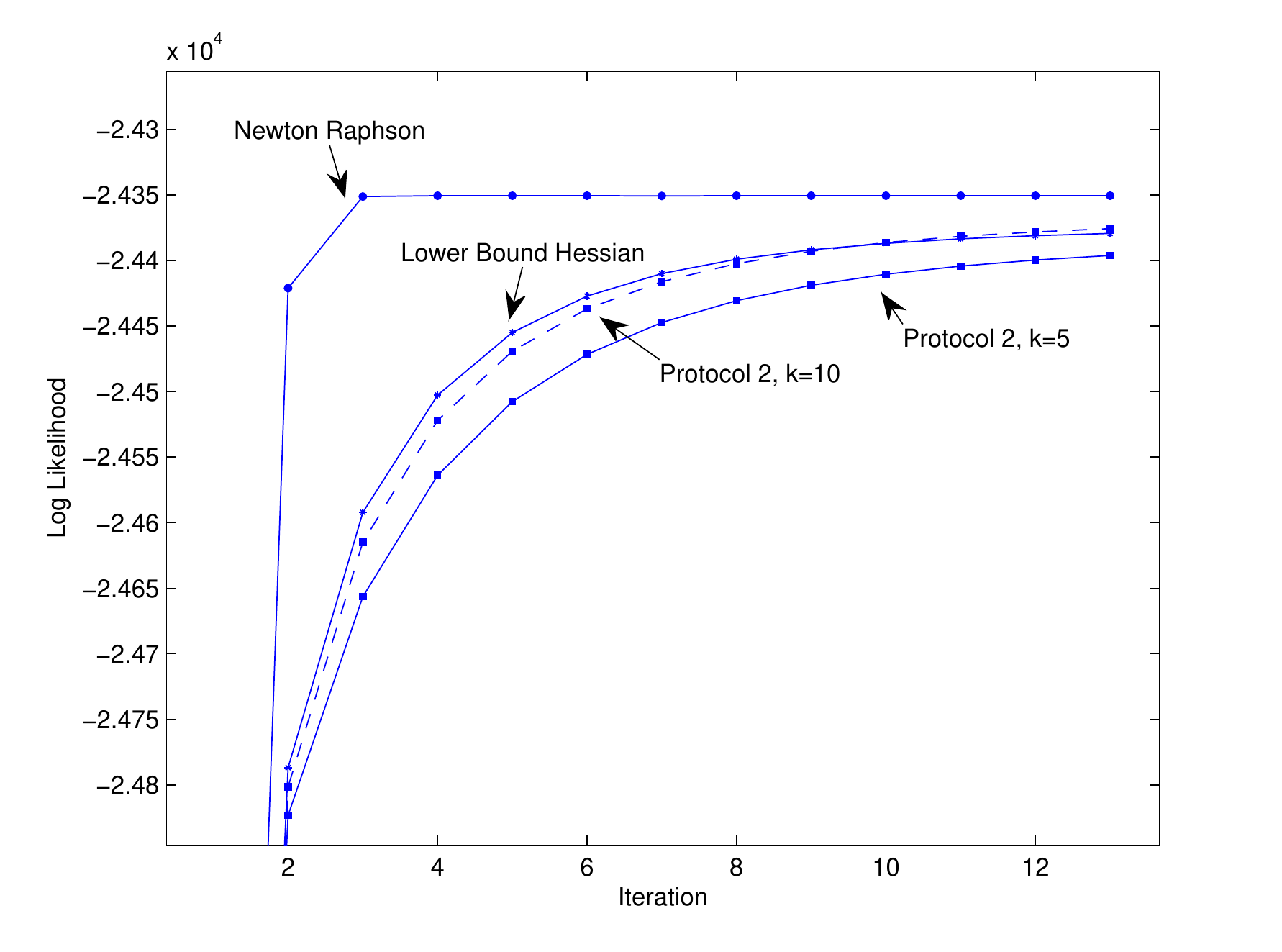}
  \caption{Log Likelihood vs iteration number for protocol 2 with $k=5,10$, and that of the ``Hessian Lower bound'' algorithm, which is the same as protocol 1 except with exact sigmoid evaluations.  We also compare to the full newton raphson method, which inverts the hessian on each iteration.}
\label{fig_like2}
\end{figure}

\section{Beyond Logistic Regression}\label{sec:extension}

We can use the construction of Section \ref{sec:protocol1} to build secure protocols for similar statistical calculations, e.g.,  the constructions for computing shares of outer products and matrix inverses naturally yield a secure algorithm for performing linear regression, for details see~\cite{fhn:10}.  Furthermore using the ``ridge regression'' penalty on the weights (i.e., computing a MAP estimate under a Gaussian prior) can naturally be added to the protocol for both linear and logistic regression.  It is also possible to implement the coordinate ascent computation of the lasso (or sparse logistic regression) using these constructions (i.e., using the GT protocol to perform soft thresholding).

Our protocol generalizes to the class of Generalized Linear Models (GLMs) including logistic regression with other link functions.  GLMs consist of a random component $Y_i$ from an exponential family, a systematic component with a linear predictor $\eta_i = x_i^T\beta$, and a link function $\eta_i=h(\mu_i)$, where $\mu_i=\mathbb{E} Y_i$. If $h$ makes the linear predictor $\eta_i=\theta_i$, where $
\theta_i$ is the natural parameter of the exponential family,  $h$ is canonical.


For Poisson log-linear models with the canonical link, $\mu_i=\exp\{\eta_i\}$, we approximate the exponential function similarly. 
For Gamma models with the canonical link, $\mu_i=1/\eta_i$, and 
for inverse-Gaussian models with the canonical link, $1/\mu^2$,  we can use the number inverting without division scheme. 
Our approach can also be extended to treat binary regression with non-canonical links, such as the \emph{probit} link function, or more generally, inverse CDF link functions. The general form of the gradient is:
\begin{equation}
\nabla \ell =  \sum_i \frac{\{y_i x_i - x_i \mu_i\}}{\text{Var}(Y_i)} \frac{\partial \mu_i}{\partial \eta_i} \; .
\end{equation}
Let $F$ denote a given CDF ($F=\Phi$ leads to the probit link function, while, of course, $F=F_L$ leads to the logit link function). Then, $\mu_i = F(\eta_i)$, and thus $\partial \mu_i/\partial \eta_i = f(\eta_i)$, where $f$ is the density. Therefore, we should find approximations for $f$ as well as for $F$ (approximation for $F$ will follow the same idea as for $F_L$, i.e., using the empirical CDF).

\section{Conclusion}

We have demonstrated that a fully secure approach to logistic regression based on the cryptographic notion of security may be made practical for use on moderately large datasets shared between several parties.  Although it is slower than methods with weaker security guarantees, it offers more rigorous guarantees with respect to the privacy of the input data.  We emphasize that our protocol (like any cryptographic protocol) prevents leakage of information which may arise from the computation itself.  It does not address any leakage which results from the output.

The problem of secure regression is far from solved however, we have yet to deal with the problem of record linkage, and have implicity assumed that the parties know how their respective datasets are aligned.  Furthermore record linkage due to a statistical model may be incorrect and may result in errorful estimates of model parameters.

\appendix
\section{Theoretical Validity of the First Protocol}\label{sec:validity_1}

Here we show how a bound on the error in the approximation (\ref{sigma-appox}) to the logistic function leads to a bound on the quality of the convergent parameter vector output by the protocol. Specifically, we establish the validity of (\ref{parm_err_bound}). Let $R$ denote a constant such that $||x_i||_2 \leq R$, for $i=1,\ldots, n$. Recall the expressions for the gradient $\nabla\ell$ and Hessian $\nabla^2\ell$ given in (\ref{grad-hess}). Define the approximated gradient, by substituting $F_L$ for $\sigma$:
\begin{equation}
\nabla\tilde{\ell}(\beta) = \sum_{i=1}^n{x_iy_i-x_iF_L(x_i^T\beta)} \; .\end{equation}
Rewriting $\nabla\ell(\beta) = \nabla\tilde{\ell}(\beta) + \sum_{i=1}^n{x_iF_L(x_i^T\beta) - x_i\sigma(x_i^T\beta)}$,
and applying the triangle inequality we obtain a bound on the norm of the gradient of the logistic objective:
\begin{equation}\label{exactnorm_bound}
\norm{\nabla\ell(\beta)}_2 \leq 
\norm{\nabla\tilde{\ell}(\beta)}_2 + nR\norm{F_L(\cdot)-\sigma(\cdot)}_\infty \; .
\end{equation}
Next we  convert a bound in the norm of the gradient into a bound on the distance to the optimum.  

\begin{lem}\label{lem_meanval}
Let $\hat{\beta}$ be the optimizer of the logistic regression objective, and let $\lambda_{\text{min}}$ denote the smallest eigenvalue of the negative Hessian in the line segment between $\beta$ and $\hat{\beta}$. Then:
\begin{equation}\label{distance_bound} 
||\beta-\hat{\beta}||_2 \leq \frac{||\nabla\ell(\beta)||_2}{\lambda_{\text{min}}}  \; .
\end{equation}
\end{lem}

\begin{proof}
We  use the mean-value theorem (for vector-valued functions) to write the difference between gradient vectors at $\beta$ and $\hat{\beta}$:
\begin{equation} \label{eq:mean-value}
\nabla\ell(\beta) - \nabla\ell(\hat{\beta}) = \nabla\ell(\beta) - 0 = \left(\int_0^1 \! \nabla^2\ell(a\beta +(1-a)\hat{\beta}) \ da \right)(\beta-\hat{\beta})  \; .
\end{equation}
Now, for every (symmetric) matrix $B$, and a non-zero vector $e$, the Rayleigh quotient satisfies $e^T Be/e^T e\geq \lambda_{\text{min}}(B)$, where $\lambda_{\text{min}}(B)$ is the minimal eigenvalue of $B$. If $B=A^2$, for a positive definite (symmetric) matrix $A$, this reduces (after taking the square root on both sides) to $\|Ae\|_2/ \|e\|_2 \geq \lambda_{\text{min}}(A)$.  Applying this to (\ref{eq:mean-value}), and using Weyl's inequality, we have:
\begin{equation} 
||\nabla\ell(\beta)||_2 = \left|\left|\left(\int_0^1 \! \nabla^2(a\beta +(1-a)\hat{\beta}) \ da \right)(\beta-\hat{\beta})\right|\right|_2 \geq \lambda_{\text{min}} ||\beta-\hat{\beta}||_2 \; .
\end{equation}
This completes the proof.
\end{proof}

\begin{lem}\label{lem_gradmin}
Using the same notation we have:
\begin{equation}\label{approxnorm_bound} 
\min_{\beta\in\mathfrak{B}} ||\nabla\tilde{\ell}(\beta)||_2 \leq nRL^{-1} \; ,
\end{equation}
where $\mathfrak{B}$ is a (non-empty) set of logistic parameters defined in the proof.
\end{lem}
\begin{proof}
Consider a continuous, monotonically non-decreasing function $g(\cdot)$ which satisfies $||g(\cdot)-F_L(\cdot)||_\infty \leq L^{-1}$.  Such a function clearly exists, for example the smooth nondecreasing curve which goes through all points $(z_{(j)},jL^{-1})$ where $1\leq j\leq L$ (where $z_{(j)}$ is $j^{th}$ smallest logistic variable used in $F_L$).  Since $g(\cdot)$ is nondecreasing, it is the derivative of some convex function:
\begin{equation}G(a) = \int_{-\infty}^a{\! g(b)\ db}\end{equation}
Consider the approximation to the logistic gradient which uses $g$ instead of $F_L$:
\begin{equation}\nabla\bar{\ell}(\beta) = \sum_{i=1}^n{x_iy_i-x_ig(x_i^T\beta)}\end{equation}
This is the derivative of a concave function:
\begin{equation}\bar{\ell}(\beta) = \sum_{i=1}^n{x_i^T\beta y_i-G(x_i^T\beta)} \; ,
\end{equation}
which is indeed concave since it is a linear function minus a convex function.  Hence $\bar{\ell}$ has a unique maximum somewhere. Consider the functions $g(\cdot)$ so that the maximum is in the interior of the space $\mathbb{R}^d$ (i.e., is not at infinity).  Hence for each such $g$ we have a point $\bar{\beta} \in \mathbb{R}^d$ where the gradient is zero, i.e., $\nabla\bar{\ell}(\bar{\beta}) = 0$. Denote the set of such $\bar{\beta}$ by $\mathfrak{B}$, and note that $\mathfrak{B}$ is not empty. 
An argument similar to the one that led to (\ref{exactnorm_bound}) shows that:
\begin{equation}
||\nabla\tilde{\ell}(\beta)||_2 = ||\nabla\bar{\ell}(\beta) + \sum_{i=1}^n{x_ig(x_i^T\beta)-x_iF_L(x_i^T\beta)}||_2 \leq ||\nabla\bar{\ell}(\beta)||_2 + nRL^{-1} \; .
\end{equation}
Therefore:
\begin{equation}||\nabla\tilde{\ell}(\bar{\beta})||_2 \leq ||\nabla\bar{\ell}(\bar{\beta})||_2 + nRL^{-1} = nRL^{-1} \; ,
\end{equation}
which completes the proof.
\end{proof}

We now put this all together and state the main result about our approximation $F_L$.

\begin{lem}
If our approximation $\nabla\tilde{\ell}$ is used as an approximation to the gradient of the logistic log likelihood, and numerical optimization is performed until $||\nabla\tilde{\ell}(\beta)||_2 \leq nRL^{-1}$, then:
\begin{equation}  ||\beta-\hat{\beta}||_2 \leq \frac{R(L^{-1} + \norm{F_L(\cdot)-\sigma(\cdot)}_\infty)}{\hat{\lambda}_{\text{min}}} \; ,\end{equation}
where $\hat{\beta}$ is the optimizer of the exact logistic regression objective, $\beta$ is the result of our numerical optimization, $R$ is the radius of a ball containing all the $x_i$, and $\hat{\lambda}_{\text{min}}$ is the smallest eigenvalue of the Fisher information matrix $I(\cdot) = -n^{-1}\nabla^2\ell(\cdot)$ in the line segment between $\beta$ and $\hat{\beta}$. \end{lem}
\begin{proof}
Notice that $||\nabla\tilde{\ell}(\beta)||_2 \leq nRL^{-1}$ is guaranteed in light of Lemma \ref{lem_gradmin}. The proof follows by substituting (\ref{exactnorm_bound}) into (\ref{distance_bound}), and by noticing that $\hat{\lambda}_{\text{min}} = n^{-1}\lambda_{\text{min}}$ and  the factors of $n$ cancel.
\end{proof}

\section{Theoretical Validity of the Coupled Iteration}\label{sec:validity_2}

Here we establish the convergence of the coupled iteration (\ref{eq:coupled}), and the error in our Taylor approximation of the logistic function.

\subsection{Monotonicity and Convergence}\label{sec:convergence}

We show that the update described in (\ref{eq:coupled}) converges monotonically towards some final value $\beta$.  We  relate the size of the step taken at one iteration  to the size of the step in the previous iteration.  We aim to show that first, these steps are always in the same directions for each unit, and secondly, the steps are monotonically decreasing and eventually the iterations converge.

\begin{lem}
$X\Delta_{t+1}$ element-wise has the same sign as $X\Delta_t$, in the sense that $X\Delta_{t+1}\circ X\Delta_t \geq 0$.
\end{lem}
\begin{proof}
If we define the idempotent matrix $M=X(X^TX)^{-1}X^T$, then we write:
\begin{eqnarray}
X\Delta_{t+1} &=& 4X(X^TX)^{-1}X(y-\hat{\sigma}_t) \nonumber \\
              &=& 4M(y-\hat{\sigma}_t)  \nonumber\\
              &=& 4M[y-\hat{\sigma}_{t-1}-(X\Delta_{t})\circ \tilde{g}_k(\hat{\sigma}_{t-1})]  \nonumber\\
              &=& 4MM(y-\hat{\sigma}_{t-1})-16M\, \text{diag}\, (\tilde{g}_k(\hat{\sigma}_{t-1}))  M(y-\hat{\sigma}_{t-1})  \nonumber\\
              &=& 4M\,\text{diag}\,(1-4\tilde{g}_k(\hat{\sigma}_{t-1}))M(y-\hat{\sigma}_{t-1})  \nonumber\\
              &=& M\,\text{diag}\,(1-4\tilde{g}_k(\hat{\sigma}_{t-1}))X\Delta_t  \; ,
\label{stepsize_rel}
\end{eqnarray}
where we made use of the idempotency of $M$.  Next considering the element-wise product as the diagonal of the outer product of these two matrices,

$$X\Delta_{t+1}(X\Delta_{t})^T = M\,\text{diag}\,(1-4\tilde{g}_k(\hat{\sigma}_{t-1}))X\Delta_t \Delta_t^TX^T$$

Since we clearly have that $1-4\tilde{g}_k(\hat{\sigma}_{t-1}) > 0$ no matter what value $\hat{\sigma}_{t-1}$ takes (due to the definition of $g_k$), we have that this matrix is the product of positive semi-definite matrices, and therefore is itself positive semi-definite.  Therefore the diagonal elements are all non-negative, and we have proved the claim.
\end{proof}

This result allows us to analyze our approximation to the logistic function as though we were using the forwards Euler method to integrate the differential equation (\ref{logistic_deriv}), since all the steps for any particular unit will be in the same direction.

\begin{lem}
As long as each step $k^{-1}|X\Delta_t| \leq  \tau < 1$ (where the inequality is element-wise), then $0 < \hat{\sigma}_t < 1,\ \forall t$ (i.e., the approximate logistic values will remain between 0 and 1).
\end{lem}
\begin{proof}
Suppose that the step is positive for all units and $\hat{\sigma}_t < 1$, then:
$$\hat{\sigma}_{t+1}-\hat{\sigma}_t \leq \tau\hat{\sigma}_t(1-\hat{\sigma}_t^2) < 1-\hat{\sigma}_t \; ,$$
so we also have that  $\hat{\sigma}_{t+1} < 1$.  Likewise for units which are involved in a negative step, if they are greater than 0, then they remain so into the next iteration by an argument which is symmetric to the one above.  Therefore we have that our logistic values never leave the interval $(0,1)$.
\end{proof}

With this we also have that $0 < 4\tilde{g}_k(\hat{\sigma}_t) < 1$ for all $t$, from the definition of $g$ and $\tilde{g}_k$.  Substitution into (\ref{stepsize_rel}), yields that:
\begin{equation}\label{approx_sigma_bound}
||X\Delta_{t+1}||_2 \leq ||M||_2\ ||\,\text{diag}\,(1-4\tilde{g}_k(\hat{\sigma}_{t-1}))||_2\ ||X\Delta_t||_2 < ||X\Delta_t||_2 \; ,
\end{equation}
since $M$ has eigenvalues which are each either 0 or 1.  This shows that the magnitude of the steps for the individual units is shrinking towards zero.  Therefore we  conclude that eventually, our approximations of the logistic values stop updating.  If we assume that $X$ has $d$ linearly independent columns, then this also implies that $\Delta_t$ is going towards zero, and therefore our algorithm eventually converges.

\subsection{Quality of the Logistic Approximation}\label{sec:p2error}

We now analyze the error in the approximation of the logistic function values.  We then use this together with  the convexity of the problem to  yield a bound on the error in the convergent parameters (see (\ref{eq:error_2})).  To aid the notation, in this section we consider the problem of estimating the logistic values for just a single case, and specifically one for which the steps are all positive.  Due to the symmetry of the logistic function about 0, we will then have the same type of bounds on the error when the approximation updates in the negative direction.  We first show a loose upper bound on the supremum of the error which would be encountered if the approximation was run for an infinite number of steps of size at most $\tau$, and then use this to bound the error after finitely many such steps.

As we have shown by the above monotonicity argument, our approximation to the logistic function is essentially analogous to using Euler's method to integrate the derivative of the logistic function.  Since we consider approximating a single value, we change the names of our variables to avoid confusion with the previous vector valued approximation.  If we denote by $\hat{s}_t$ the approximated value after $t$ steps of various sizes,  $\tau_0\ldots \tau_{t-1} < \tau$.  Thus  $\hat{s}_t \approx s_t = \sigma(a_t)$ where $a_t = \sum_{i=0}^{t-1}\tau_i$. We  compare this approximation to the exact values and consider the error:
$$\xi_t = \hat{s}_{t} - s_{t} \; .$$

Making use of the step (\ref{sigma_lin}), we  evaluate the error in the next iteration:
\begin{align*}
\xi_{t+1} & = \hat{s}_{t+1} - s_{t+1} \\
          & = \hat{s}_t + \tau_{t}g(\hat{s}_t) -s_t -\tau_tg(s_t) - 2^{-1}\tau_t\sigma^{\prime\prime}(\cdot)|_{a^\star_t} \\
          & = \xi_t + \tau_t[g(\hat{s}_t) - g(s_t)] + \zeta_t \\
          & = \xi_t + \tau_t(\hat{s}_t-s_t)g^\prime(\cdot)|_{s^{\star}_t} + \zeta_t \\
          & = \xi_t(1+\tau_tg^\prime(\cdot)\big|_{s^{\star}_t}) + \zeta_t \\
          & = \xi_t(1+\tau_t -2\tau_ts^{\star}_t) + \zeta_t
\end{align*}
Where we have defined
\begin{equation}\label{eqn:zeta}\zeta_t = - 2^{-1}\tau_t\sigma^{\prime\prime}(\cdot)\big|_{a^\star_t} \end{equation}

\noindent and $a_t \leq a^\star_t \leq a_{t+1}$ is some value in the interval about which the second derivative is taken.  Likewise $s^\star_t$ is bounded between $s_t$ and $\hat{s}_t$.  As we have seen from (\ref{approx_sigma_bound}), as long as $\tau_t \leq \tau < 1$ then $0 < \hat{s}_t < 1$ for all $t$.  Since we only consider positive steps $\tau_t > 0$ then we have that $2^{-1} \leq \hat{s}_t < 1$, and hence the same bound applies to $s_t^\star$.  Therefore we have that:
$$|\xi_{t+1}| \leq |\xi_t| + |\zeta_t|$$
Therefore we see that:
\begin{equation}\label{bound_sum}\sup_t|\xi_t| \leq \sum_{i=1}^\infty |\zeta_i|\end{equation}

Examining the form of $\sigma^{\prime\prime}(\cdot)$, we find it to be a function which is everywhere negative.  Examining the third derivative, we find that the second derivative has exactly one stationary point in $[0,\infty)$ which is located at:

$$a^\star = -\log{\frac{6-\sqrt{12}}{6+\sqrt{12}}}, \quad \sigma(a^\star) = \frac{6+\sqrt{12}}{12} $$

Whats more, we see that $\partial^3\sigma(\cdot) < 0$ on $[0,a^\star)$, and $\partial^3\sigma(\cdot) > 0$ on $(a^\star,\infty)$.  Therefore we have that $a^\star$ is the minimum of the function.  Using this we  bound the sum (\ref{bound_sum}) by an integral:

$$-\sum_{t=0}^\infty\sigma^{\prime\prime}(\cdot)|_{a^\star_t} \leq -\int_0^\infty \sigma^{\prime\prime}(a)\ da - 2\tau\sigma^{\prime\prime}(x^\star)
        = 4^{-1} - 2\tau\sigma^{\prime\prime}(x^\star)
\; .
$$
Substituting this into (\ref{bound_sum}) and (\ref{eqn:zeta}) we have that:
\begin{equation}\label{euler_bound}\max_t|\xi_t| \leq 2^{-1}\tau(4^{-1} - 2\tau\sigma^{\prime\prime}(x^\star)) \stackrel{\text{def}}{=} c\tau + d\tau^2 \approx c\tau \; .
\end{equation}
We can make  the approximation arbitrarily tight by decreasing the step size.






\end{document}